\documentclass[12pt]{article}
\usepackage{amsfonts}
\usepackage{amssymb}
\usepackage{graphicx}
\usepackage{amsmath}
\usepackage{natbib}
\usepackage{setspace}
\usepackage{amsthm}
\usepackage[margin=1in]{geometry}
\usepackage[hidelinks]{hyperref}

\newtheorem{theorem}{Theorem}

\newtheorem{proposition}{Proposition}
\theoremstyle{definition}
\newtheorem{definition}{Definition}
\newtheorem{assumption}{Assumption}
\theoremstyle{remark}

\newtheorem{example}{Example}

\DeclareMathOperator{\E}{E}
\input{tcilatex}
\begin{document}

\title{Optimal Pricing of Cloud Services: \\ Committed Spend under Demand
Uncertainty\thanks{%
Dirk Bergemann gratefully acknowledges financial support from NSF SES
2049754 and ONR MURI. We thank Scott Shenker for productive conversations.}}
\author{Dirk Bergemann\thanks{%
Department of Economics, Yale University, New Haven, CT 06511,
dirk.bergemann@yale.edu} \and Michael C. Wang\thanks{%
Department of Economics, Yale University, New Haven, CT 06511,
michael.wang.mcw75@yale.edu}}
\date{\today}
\maketitle

\begin{abstract}
We consider a seller who offers services to a buyer with multi-unit demand.
Prior to the realization of demand, the buyer receives a noisy signal of
their future demand, and the seller can design contracts based on the
reported value of this signal. Thus, the buyer can contract with the service
provider for an unknown level of future consumption, such as in the market
for cloud computing resources or software services. We characterize the
optimal dynamic contract, extending the classic sequential screening
framework to a nonlinear and multi-unit setting. The optimal mechanism
gives discounts to buyers who report higher signals, but in exchange they
must provide larger fixed payments. We then describe how the optimal
mechanism can be implemented by two common forms of contracts observed in
practice, the two-part tariff and the committed spend contract. Finally, we
use extensions of our base model to shed light on policy-focused questions,
such as analyzing how the optimal contract changes when the buyer faces
commitment costs, or when there are liquid spot markets.

\medskip \noindent \textsc{Keywords:} Sequential Screening, Second-Degree
Price Discrimination, Mechanism Design, Cloud Computing, Commitment Spend
Contract, Software as a Service Contract

\medskip

\noindent \textsc{JEL Classification:} D44, D82, D83.
\end{abstract}

\newpage

\section{Introduction}

\subsection{Motivation}

The digital economy increasingly runs on service contracts that commit
buyers to future purchases before they fully understand their needs. In 
\emph{committed spend} contracts, cloud computing providers like Amazon Web
Services, Microsoft Azure, and Google Cloud Platform offer substantial
discounts to customers who commit to minimum spending levels over a time
period of variable length. Similar arrangements appear in the enterprise
software licensing, data services, and API access contracts. These committed
spend agreements have attracted regulatory scrutiny, with the UK's
Competition and Markets Authority launching an investigation in 2023 over
concerns that they may reduce competition by locking customers into single
providers.

A prominent example of service contracts are contracts for software as a
service (SaaS) and cloud computing. These contracts often feature committed
spend agreements. ``Those are agreements between a cloud provider and a
customer in which the customer commits to spend a minimum amount across the
cloud provider's cloud services over a period of years, and in return,
receives a percentage discount on its spend with that provider during those
same years'' (\cite{coma24}). A third prominent example are service
agreements for data and LLMs that cover the use of data for training and the
use of the model and weights, respectively.

Yet the theoretical foundations for such contracts remain poorly understood.
While the basic economic intuition is clear---providers offer discounts in
exchange for reduced demand uncertainty---the optimal structure of these
contracts involves subtle trade-offs. Buyers receive noisy signals about
their future demand, but face genuine uncertainty when signing contracts.
Contracts must balance the benefits of early commitment against the costs of
potential misallocation. Moreover, real-world frictions like capital
constraints and spot market alternatives shape both the feasible and optimal
contract forms.

This paper provides a comprehensive theoretical analysis of optimal service
contracting under demand uncertainty. We consider a monopolist offering a
service or portfolio of services to a buyer over a certain time period. The
buyer can choose a level (or intensity) of the service to use. This level
can refer to either the quantity or the quality of services used.

In the initial period, the users have imperfect information about their
willingness to pay for the service, and thus they have an expectation over
the use and their willingness to pay, but do not know for sure. We model
this with the buyers having a prior distribution over the willingness to
pay. The users contract on the basis of knowing their prior estimate, but
acknowledging that their eventual value when buying is given by a realized
demand level. The seller can offer contracts in the initial period that
specify both upfront payments and future usage-based prices. As the
willingness to pay evolves over time, the contract, either implicitly or
explicitly, involves an element of sequential screening. Furthermore, given
the variable level of the good provided, the provider simultaneously employs
the tools of second degree price discrimination.

The first objective of the paper is to derive the optimal contract
(mechanism) to solve the allocation problem as a revenue maximizing solution
for the service provider. Our first main result characterizes the optimal
dynamic mechanism. We show that buyers who report higher signals should
receive more favorable usage rates, but must provide larger commitments.
This matches the structure of real-world cloud computing contracts, where
larger commitments come with higher discounts. The optimal mechanism can be
implemented in two ways that mirror common practice: (1) A two-part tariff,
which consists of an upfront fee combined with usage-based prices, where
higher upfront fees lead to lower usage payments, and (2) Committed spend
contracts, which combine a minimum purchase requirement with discounted
rates, where larger commitments yield greater discounts. Our results provide
both theoretical foundations for existing practices and practical guidance
for contract design in digital service markets.

These implementations are revenue-equivalent under standard conditions but
diverge when we introduce realistic frictions. Our second set of results
analyzes two key frictions: First, we consider the case of capital
constraints. When buyers face high costs of capital (making early payments
especially costly), pure committed spend contracts become optimal. This
helps explain why cloud providers often use committed spend rather than
upfront payments, particularly with startup customers. Second, we consider
spot market alternatives: When buyers have access to competitive spot
markets, the optimal contract offers larger discounts to high-demand buyers
while potentially excluding or distorting service for low-demand buyers.
This creates a form of market segmentation between committed and spot market
customers.

\subsection{Related Literature}

We offer a generalization of the sequential screening problem of \cite%
{coli00} to a continuous demand model rather than single-unit demand model.
The continuous demand model used is a version of \cite{muro78} and \cite%
{mari84}. A general class of socially efficient dynamic mechanisms is
provided by \cite{beva10} and \cite{atse13}. A related class of allocation
problems were analyzed under the objective of revenue maximization by \cite%
{past14} and \cite{best15}. A recent survey of this literature appears in 
\cite{beva19}. Given the presence of sequential screening, the optimal
(revenue maximizing) contract includes elements of an option, namely an
option fee and a strike price, that is a commitment price and usage price,
see \cite{best22}. A broader survey of the economics of cloud computing is
offered in \cite{bicm24}.

This paper is also related to a more general literature on mechanism design
with participation constraints stronger than those considered in the classic
work of \cite{myer81}. The structure of optimal mechanisms with
type-dependent outside options is considered in \cite{lesa89}. The
implications of stronger participation constraints specifically on the
optimal sequential screening contract are analyzed in \cite{krst15a} and 
\cite{becw20}.

\section{Model}

We consider a buyer and seller who operate in two periods $\tau \in \{0,1\}$%
. Trade occurs in period 1: the seller produces some quantity $q\geq 0$ of a
good at constant marginal cost $c\geq 0$, which is purchased by the buyer.
The buyer's utility is quasilinear in the transfer $t$, and displays
constant demand elasticity in the quantity $q$: 
\begin{equation*}
u(v,q,t)=vq^{\alpha }-t,\quad \alpha \in (0,1).
\end{equation*}%
The seller's payoff is the transfer minus their cost of production: 
\begin{equation*}
\Pi (q,t)=t-qc.
\end{equation*}%
In period 0, the buyer observes a one-dimensional signal $\theta \in \lbrack 
\underline{\theta },\overline{\theta }]\subset \mathbb{R}_{+}$ of their
value. Let $G(v\mid \theta )$ denote the distribution of $v$ conditional on $%
\theta $, and $F(\theta )$ the distribution of the initial signal. We assume
that the distribution of values is increasing in the signal $\theta $, in
the stochastic dominance sense:%
\begin{equation*}
G(v\mid \theta )<G(v\mid \theta ^{\prime }),\ \forall v,\forall \theta
^{\prime }<\theta \text{.}
\end{equation*}

\begin{assumption}[First Order Stochastic Domiance Ordering]
\label{asu:fosd} \qquad \newline
$G(\cdot \mid \theta )$ is ordered in first order stochastic dominance by $%
\theta $.
\end{assumption}

We think of $\theta$ as a parameter describing the firm's ``demand
projection''; it does not enter the firm's actual payoff, but provides
information about what the firm's realized demand in the next period will
be. Larger $\theta$ corresponds to more optimistic demands, which we capture
in the FOSD ordering of $G$.

To derive the optimal mechanism, we will need a form of regularity. Define
the \emph{dynamic virtual value} as 
\begin{equation}
\varphi (\theta ,v)=v+\frac{1-F(\theta )}{f(\theta )}\cdot \frac{\partial
G(v\mid \theta )/\partial \theta }{g(v\mid \theta )}.  \label{eq:dvv}
\end{equation}

\begin{assumption}[Regularity]
\label{asu:reg} \qquad \newline
$\varphi(\theta,v)$ is weakly increasing in both $\theta$ and $v$.
\end{assumption}

Assumptions \ref{asu:fosd}-\ref{asu:reg} allow us to apply the insights of
the earlier cited literature on sequential screening and dynamic mechanism
design, and focus on how to use these results to study the design of
committed spend and service contracts.

In this setting, a \emph{mechanism} (or contract) $\mathcal{M}:=(q,t)$
specifies, for each signal and value realization pair $(\theta ,v)$, a
quantity $q(\theta ,v)$ and transfer $t(\theta ,v)$. By the revelation
principle, it is without loss of generality to restrict attention to
mechanisms which are incentive compatible, meaning the buyer finds it
optimal to truthfully report $\theta $, and then $v$.

The seller's problem is thus to choose $q$ and $t$ to maximize expected
profit: 
\begin{equation*}
\max_{q,t}\ \E_{\theta \times v}\Big[t(\theta ,v)-q(\theta ,v)c\Big]
\end{equation*}%
subject to incentive compatibility (IC) in period $0$ and period $1:$ 
\begin{align}
& \E_{v}\Big[u(v,q(\theta ,v),t(\theta ,v)\big)\mid \theta \Big]\geq \E_{v}%
\Big[u\big(v,q(\theta ^{\prime },v),t(\theta ^{\prime },v)\big)\mid \theta %
\Big]\ \forall \theta ,\theta ^{\prime };  \tag{IC0}  \label{eq:IC0} \\
& u\big(v,q(\theta ,v),t(\theta ,v)\big)\geq u\big(v,q(\theta ,v^{\prime
}),t(\theta ,v^{\prime })\big)\ \forall \theta ,v,v^{\prime };  \tag{IC1}
\label{eq:IC1}
\end{align}%
and interim individual rationality in period $0:$%
\begin{equation}
\E_{v}\Big[u\big(v,q(\theta ,v),t(\theta ,v)\big)\mid \theta \Big]\geq 0\
\forall \theta \text{.}  \tag{IR}  \label{eq:IR}
\end{equation}%
Note that \eqref{eq:IC0} contains only a subset of the period-0 IC
constraints, namely deviation with respect to the initial report, but the
remainder of the deviations are automatically satisfied if \eqref{eq:IC1} is.

We believe this framework describes markets in which, prior to the actual
delivery of the services, buyers and seller establish future pricing based
on the buyer's expected demand, but this expected demand has no effect on
the seller's marginal costs. Our leading example is the cloud computing
service market, where buyers receive per-unit discounts for committing to
certain levels of use, although the setup just as easily describes many
service and digital resource allocation contracts. This model could also be
applied to physical goods, assuming that the marginal costs associated with
serving high- and low-demand customers are the same; note that the
solicitation of the demand forecast $\theta$ is used purely to screen
customers, and does not affect the choice of production technology.

\section{Optimal Dynamic Mechanism\label{sec:odm}}

We first determine the optimal dynamic direct mechanism. Then we consider
indirect implementations of the direct mechanism in form of a two-part
tariff and a committed spend contract.

\subsection{Dynamic Direct Mechanism\label{subsec:ddm}}

Following \cite{coli00} and \cite{past14}, the optimal mechanism can be
found with a Myersonian approach. For notational simplicity, denote by 
\begin{equation*}
u(\theta ,v):=u(\theta ,q(\theta ,v),t(\theta ,v)).
\end{equation*}
In period 1, the mechanism devolves to a standard static mechanism. That is, %
\eqref{eq:IC1} requires that 
\begin{equation*}
\frac{\partial }{\partial v}u(\theta ,v)=q(\theta ,v)^{\alpha }\implies
u(\theta ,v)=u(\theta ,\underline{v})+\int_{\underline{v}}^{v}q(\theta
,x)^{\alpha }\ dx.
\end{equation*}%
Using this, we can derive the first-period envelope condition: 
\begin{equation*}
\frac{d}{d\theta }\E_{v}\big[u(\theta ,v)\mid \theta \big]=\int_{\underline{v%
}}^{\overline{v}}u(\theta ,v)\frac{\partial g(v\mid \theta )}{\partial
\theta }\ dv=-\int_{\underline{v}}^{\overline{v}}q(\theta ,v)^{\alpha }\frac{%
\partial G(v\mid \theta )}{\partial \theta }\ dv.
\end{equation*}%
Let us consider the relaxed seller's problem which only imposes the local
\eqref{eq:IC0} constraints (along with \eqref{eq:IR} and all \eqref{eq:IC1} constraints). The objective
function in this relaxed problem, from standard mechanism design techniques,
is 
\begin{equation*}
\E_{\theta ,v}\Big[\varphi (\theta ,v)q^{\alpha }-cq\Big]-\E_{v}\Big[u(%
\theta,v)\mid \theta =\underline{\theta }\Big].
\end{equation*}%
Setting the IR constraint to be binding, the point-wise maximization is 
\begin{equation}
q^{\ast }(\theta ,v)=\left( \frac{\alpha }{c}\varphi (\theta ,v)\mathbb{I}%
_{[\varphi (\theta ,v)\geq 0]}\right) ^{\frac{1}{1-\alpha }},
\label{eq:opt-q}
\end{equation}%
where the indicator function $\mathbb{I}_{[\varphi (\theta ,v)\geq 0]}$
indicates that only customers with a positive virtual utility, $\varphi
(\theta ,v)$ as defined earlier in (\ref{eq:dvv}), are offered a positive
level of service:%
\begin{equation*}
\mathbb{I}_{[\varphi (\theta ,v)\geq 0]}=\left\{ 
\begin{array}{cc}
0, & \text{if }\varphi (\theta ,v)<0; \\ 
1, & \text{if }\varphi (\theta ,v)\geq 0.%
\end{array}%
\right.
\end{equation*}
The transfers are pinned down by the envelope conditions from before: 
\begin{align}
& u(\theta ,v)=u(\theta ,\underline{v})+\int_{\underline{v}}^{v}q^{\ast
}(\theta ,x)^{\alpha }\ dx,  \label{eq:u1} \\
& \E_{v}\Big[u(\theta ,v)\mid \theta \Big]=-\int_{\underline{\theta }%
}^{\theta }\int_{\underline{v}}^{\overline{v}}q^{\ast }(x,v)^{\alpha }\frac{%
\partial G(v\mid x)}{\partial x}\ dv\ dx.  \label{eq:opt-ir}
\end{align}

Under the regularity conditions imposed earlier, this mechanism satisfies
the global \eqref{eq:IC0} constraints, and hence is optimal.

\begin{theorem}[Optimal Dynamic Direct Mechanism]
\qquad \newline
Under Assumptions \ref{asu:fosd}-\ref{asu:reg}, the mechanism characterized
by \eqref{eq:opt-q}-\eqref{eq:opt-ir} is optimal.
\end{theorem}

\begin{proof}
The construction of this mechanism satisfies \eqref{eq:IR} and \eqref{eq:IC1}%
, so all we need is to show that the global \eqref{eq:IC0} constraints are
satisfied. We do this by showing that the indirect utility function 
\begin{equation*}
w(\theta,\theta^\prime) = \int_{\underline{v}}^{\overline{v}} \big(v
q^\ast(\theta^\prime,v) - t(\theta^\prime,v)\big) g(v \mid \theta)\ dv
\end{equation*}
satisfies the single-crossing property over $\theta$ and $\theta^\prime$
whenever $q(\theta^\prime,v)$ is increasing (pointwise) in $\theta^\prime$,
which is implied by Assumption \ref{asu:reg}.\footnote{%
We could apply \cite{past14} Corollary 1 of Theorem 3, but for completeness
we provide a direct proof.} In particular, we wish to show that 
\begin{equation*}
w(\theta,\theta_H) - w(\theta,\theta_L)
\end{equation*}
is increasing in $\theta$ for any $\theta_L \leq \theta_H$. For every $%
\theta^\prime$, the static mechanism $\{q(\theta^\prime,v),
t(\theta^\prime,v)\}$ is incentive compatible, which in particular by
standard results means that 
\begin{equation*}
v q^\ast(\theta^\prime,v) - t(\theta^\prime,v) = u(\theta^\prime,\underline{v%
}) + \int_{\underline{v}}^v q^\ast(\theta^\prime,x)^\alpha\ dx
\end{equation*}
for all $\theta$. That is, after reporting $\theta^\prime$ initially, the
payoff after $v$ realizes ex-post in period 1 is the same regardless of the
true $\theta$. Thus, we can decompose 
\begin{equation*}
w(\theta,\theta_H) - w(\theta,\theta_L) = \big(u(\theta_H,\underline{v}) -
u(\theta_L,\underline{v})\big) + \int_{\underline{v}}^{\overline{v}}
\left(\int_{\underline{v}}^v q^\ast(\theta_H,x)^\alpha -
q^\ast(\theta_L,x)^\alpha\ dx\right) g(v \mid \theta)\ dv.
\end{equation*}
$q(\theta^\prime,v)$ is increasing in $\theta^\prime$, so the inner integral
(the term in parentheses) is increasing in $v$. Thus, since $G(\cdot \mid
\theta)$ are ordered in FOSD, this means that the entire expression is
increasing in $\theta$.
\end{proof}

The optimal mechanism delivers, for every realized value of demand $v$,
higher quality to buyers with higher first-period signals $\theta$. However,
for low $v$, buyers with higher $\theta$ receive lower net utility. That is,
buyers who report more optimistic demand projections are penalized for
having low actual demand, but rewarded (relative to reporting a low demand
projection) when their realizations are high.

\begin{example}\label{ex:1}
We now introduce our running example for the remainder of the paper. Suppose that $v = \theta z$, where $z \sim U[\frac{1}{2},1]$ and $\theta
\sim U[1,2]$. Then, 
\begin{equation*}
\varphi(\theta,v) = 2z(\theta-1) = \frac{2v(\theta-1)}{\theta} \implies 
\frac{dt}{dq} = \frac{\theta}{2(\theta-1)}c.
\end{equation*}
The optimal mechanism for selected values of $\theta$ is plotted in Figure \ref{fig:opt-m}. Quantity provided, as function of the realized demand $v$, is plotted with solid lines, while total transfers are plotted with dashed lines.

\begin{figure}[ht]
    \centering
    \includegraphics[height=0.4\textwidth]{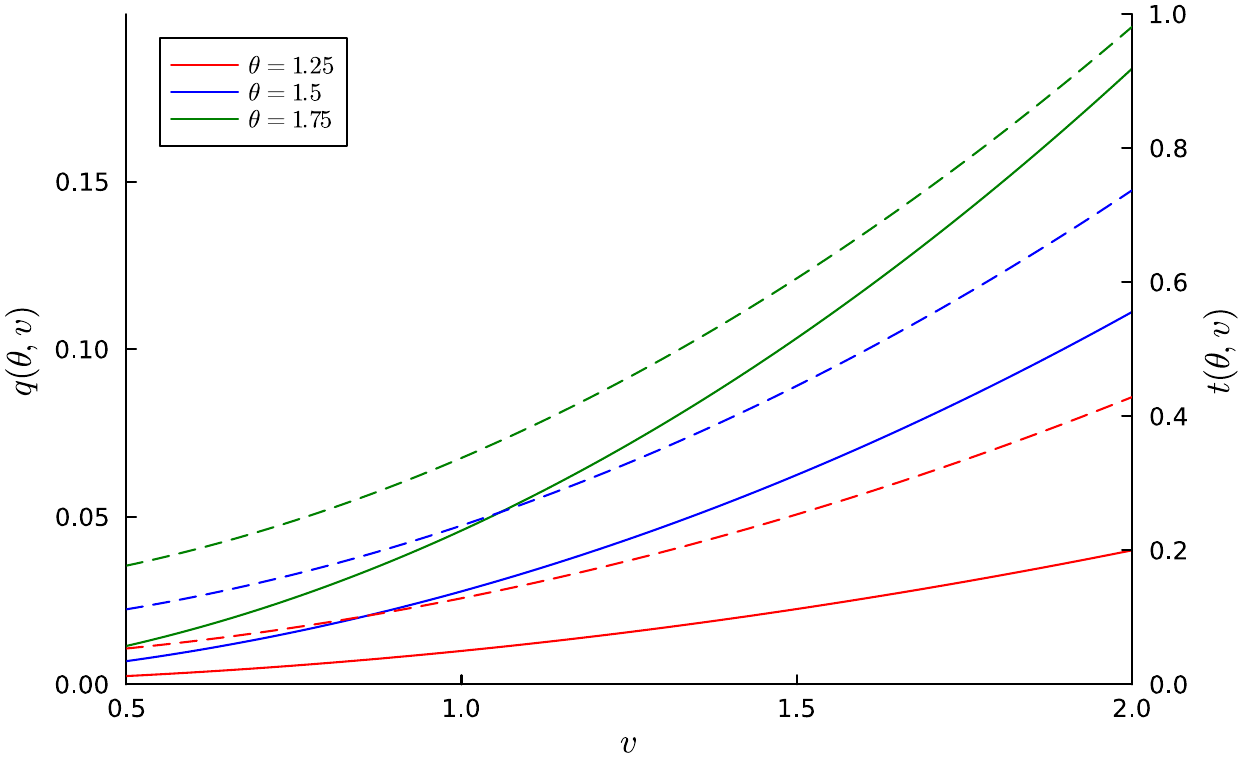}
    \caption{Optimal Mechanism for Example \ref{ex:1} ($q$ solid, $t$ dashed).}\label{fig:opt-m}
\end{figure}
\end{example}

\subsection{Two-Part Tariff\label{subsec:two}}

We now describe two natural implementations of the optimal direct mechanism
via indirect mechanisms which match contracts observed in reality. In this
section, we allow the mechanism to separate transfers across the two
periods, i.e.\ the contract specifies payments $t_{0}(\theta )$ and $%
t_{1}(\theta ,v)$ after periods 0 and 1, respectively. For now, we assume
the buyer and seller have no preference for payment timing, so that any $%
t_{0},t_{1}$ such that 
\begin{equation*}
t_{0}(\theta )+t_{1}(\theta ,v)=t(\theta ,v)\ \forall \theta ,v
\end{equation*}%
for some $t$ are revenue-equivalent.

For now, what follows are simply particular selections of ways to divide up
the transfers of the optimal mechanism; later, when we add in additional
design considerations, the choice of timing will matter.

A \emph{(non-linear) pricing schedule} is a function $p: \mathbb{R}_+
\rightarrow \mathbb{R}_+$ which specifies a set of possible quantities and a
price for each quantity. Observe that by definition, the price for each of
these quantities must be positive, i.e.\ this rules out mechanisms which
involve refunds from the seller to the buyer.\footnote{%
Such refunds might be desirable if the buyer has imperfect commitment power.
Then, charging large prices up-front and providing a refund in the second
period alleviates the lack of commitment. However, we rarely observe refunds
in practice.}

\begin{definition}[Two-Part Tariff]
\qquad \newline
A \emph{two-part tariff} is a contract consisting of an initial payment
function $t_0: [\underline{\theta},\overline{\theta}] \rightarrow \mathbb{R}%
_+$ and a collection of (type-dependent) pricing schedules $\{p_\theta\}$
such that for every $\theta$, 
\begin{equation*}
\min_{q \in \mathbb{R}_+} \big[p_\theta(q)\big] = 0.
\end{equation*}
That is, the buyer can always decline to make a payment in the second period.
\end{definition}

Let 
\begin{equation*}
\underline{v}(\theta ):=\inf \{v\mid g(v\mid \theta )>0\}
\end{equation*}%
denote the low end of the support of $G(\cdot \mid \theta )$. For any $%
\theta $ such that 
\begin{equation}
\varphi (\theta ,\underline{v}(\theta ))=0,  \label{eq:unique}
\end{equation}%
there is a unique\footnote{%
Up to choice of $p_{\theta }(q)$ which are never chosen in equilibrium.}
two-part tariff implementing the optimal mechanism. Otherwise, there is
indeterminacy in how the seller wants to split $t(\theta ,\underline{v}%
(\theta ))$ into the up-front payment and the price schedule. When %
\eqref{eq:unique} holds, we set 
\begin{equation}
t_{0}(\theta )=-u(\theta ,\underline{v})\text{ and }t_{1}(\theta
,v)=t(\theta ,v)+u(\theta ,\underline{v}).  \label{eq:upfront}
\end{equation}%
The non-linear price schedule for type $\theta $ is defined by: 
\begin{equation*}
p_{\theta }(q^{\ast }(\theta ,v))=t_{1}(\theta ,v).
\end{equation*}%
Since $u(\theta ,v)$ is increasing in $v$ by \eqref{eq:IC1}, and $p(q^{\ast
}(\theta ,\underline{v}))=0$, the construction clearly produces a valid
pricing schedule $p_{\theta }$ for every $\theta $. What remains is to check
that the initial payments $-u(\theta ,\underline{v})$ are all positive. By %
\eqref{eq:IR}, 
\begin{equation*}
\E_{v}\big[u(\underline{\theta },v)\big]=0\implies u(\underline{\theta },%
\underline{v})\leq 0,
\end{equation*}%
and by the period-0 envelope condition \eqref{eq:opt-ir}, 
\begin{equation*}
t_{0}(\theta )=\int_{\underline{v}}^{\overline{v}}\int_{\underline{v}%
}^{v}\left( \frac{\alpha }{c}\varphi (\theta ,x)\right) ^{\frac{\alpha }{%
1-\alpha }}\ dx\ g(v\mid \theta )\ dv+\int_{\underline{\theta }}^{\theta
}\int_{\underline{v}}^{\overline{v}}\left( \frac{\alpha }{c}\varphi
(x,v)\right) ^{\frac{\alpha }{1-\alpha }}\frac{\partial G(v\mid x)}{\partial
x}\ dv\ dx.
\end{equation*}%
Doing the usual integration by parts, the first term is equal to 
\begin{equation*}
\int_{\underline{v}}^{\overline{v}}\left( \frac{\alpha }{c}\varphi (\theta
,v)\right) ^{\frac{\alpha }{1-\alpha }}(1-G(v\mid \theta ))\ dv.
\end{equation*}%
Finally, we can differentiate to get 
\begin{equation*}
\frac{\partial }{\partial \theta }\Big[-u(\theta ,\underline{v})\Big]=\left( 
\frac{\alpha }{c}\right) ^{\frac{\alpha }{1-\alpha }}\int_{\underline{v}}^{%
\overline{v}}\frac{\alpha }{1-\alpha }\varphi (\theta ,v)^{\frac{2\alpha -1}{%
1-\alpha }}\left[ \frac{\partial }{\partial \theta }\varphi (\theta ,v)%
\right] (1-G(v\mid \theta ))\ dv\geq 0,
\end{equation*}%
showing that the initial payment is always positive. In fact, it is
increasing in $\theta $, which aligns with how these contracts look in
practice.

Observe that we can write 
\begin{equation*}
t_{1}(\theta ,v)=t(\theta ,v)+u(\theta ,\underline{v})=\left( \frac{\alpha }{%
c}\right) ^{\frac{\alpha }{1-\alpha }}\left( v\varphi (\theta ,v)^{\frac{%
\alpha }{1-\alpha }}-\int_{\underline{v}}^{v}\varphi (\theta ,x)^{\frac{%
\alpha }{1-\alpha }}\ dx\right) .
\end{equation*}%
Using implicit differentiation, we can compute the marginal price of
quantity: 
\begin{equation*}
\frac{dt}{dq}=\frac{\partial t(\theta ,v)/\partial v}{\partial q(\theta
,v)/\partial v}=\frac{cv}{\varphi (\theta ,v)}.
\end{equation*}%
From above, we can see that if 
\begin{equation*}
\frac{1}{v}\cdot \frac{\partial G(v\mid \theta )/\partial \theta }{g(v\mid
\theta )},
\end{equation*}%
is increasing in $v$ (note that the expression is negative, so decreasing in
absolute value), then the marginal price of quantity is decreasing in $v$. A
sufficient condition for this is that 
\begin{equation*}
\left\vert \frac{\partial G(v\mid \theta )/\partial \theta }{g(v\mid \theta )%
}\right\vert
\end{equation*}%
is decreasing in $v$, which is a form of a monotone hazard rate condition.

This two-part tariff mirrors contracts in which the buyer incurs some fixed
contracting fee, and then at the time of service realization faces a
non-linear pricing schedule. The pricing schedule becomes more generous when
they sign ``larger'' service contracts, where larger is both in the sense of
the fixed fee and the expected quantity of service provided.

A setting in which the optimal mechanism is particularly tractable is the
multiplicative values setting, in which $v=\theta z$, and $z\sim H$ is
independent of $\theta \sim F$. Then, 
\begin{equation*}
\varphi (\theta ,v)=\theta z-\frac{1-F(\theta )}{f(\theta )}z=z\varphi
_{F}(\theta ).
\end{equation*}%
Here, $\varphi _{F}$ is the usual one-dimensional virtual value. Assumption %
\ref{asu:reg} becomes the condition that 
\begin{equation*}
\varphi (\theta ,v)=v\cdot \frac{\varphi _{F}(\theta )}{\theta }
\end{equation*}%
is increasing in both $v$ and $\theta $ (when positive), which is true if
and only if $F$ satisfies monotone hazard rate, i.e.\ that 
\begin{equation*}
\frac{1-F(\theta )}{f(\theta )}
\end{equation*}%
is decreasing in $\theta $. The optimal contract can be implemented with an
up-front payment increasing in $\theta $ and a constant unit price of 
\begin{equation*}
\frac{\theta }{\varphi _{F}(\theta )}c,
\end{equation*}%
or, equivalently, a constant mark-up of 
\begin{equation*}
\frac{\theta }{\varphi _{F}(\theta )}-1.
\end{equation*}%
In this case, we see clearly that larger $\theta $ get more advantageous
prices in the second-stage, at the cost of a higher up-front payment.

The notable feature of this setting is that the second stage pricing
schedule is in fact linear, with a unit price which is independent of $H$.
This setting is related to that of \cite{best15}, who find that a geometric
Brownian motion structure simplifies incentive compatibility constraints in
the continuous dynamic mechanism design context. We work in the
multiplicative setting for Section \ref{sec:spot-market}, our extension
analyzing the optimal dynamic contract with spot markets.

Recall Example \ref{ex:1}, where we assumed $v = \theta z$, $z \sim U[\frac{1}{2},1]$, and $\theta
\sim U[1,2]$. Thus, the optimal contract can be implemented as a two-part tariff where in
period 1, the buyer faces a constant unit price of $\frac{2\theta}{\theta - 1%
}c$. In Figure \ref{fig:1}, we plot the up-front payment when $c = 1$ and $%
\alpha = \frac{1}{2}$, along with the corresponding unit prices associated with each $\theta$ (the parameters chosen make the scale of the two quite different, so they are plotted on parallel $y$ axes).

\begin{figure}[ht]
\centering
\includegraphics[height=0.4\textwidth]{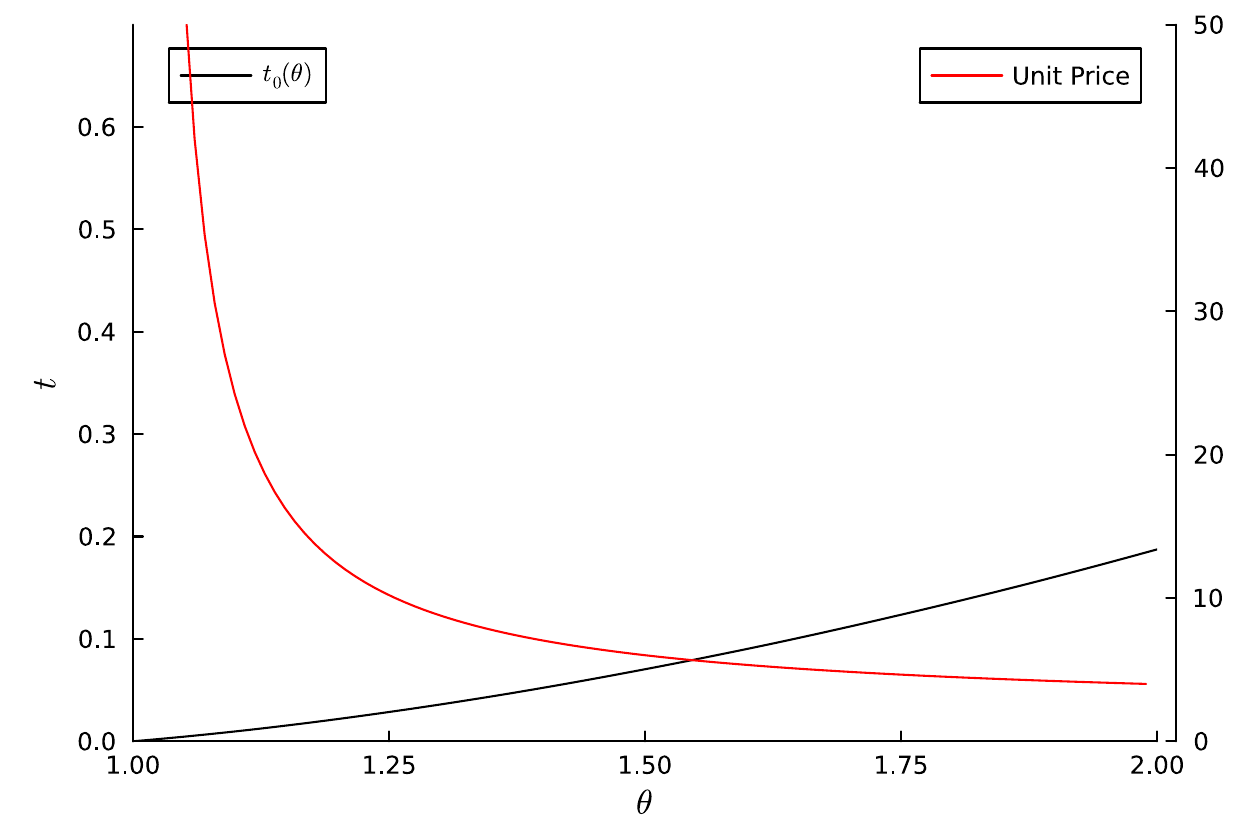}
\caption{$t_0(\protect\theta)$ and Unit Price for Example \ref{ex:1} when $c=1$ and $\protect\alpha %
= \frac{1}{2}$.}
\label{fig:1}
\end{figure}

In this example, because $\varphi(\theta,\underline{v}(\theta)) > 0$ for all 
$\theta$, the optimal two-part tariff is not unique. The up-front
payment are computed as in \eqref{eq:upfront}.

\subsection{Committed Spend Contract}

Another type of contract commonly observed in practice, which has recently
attracted significant regulatory attention, are contracts which do not
extract payments upfront but specify mandatory minimum spends.

\begin{definition}[Committed Spend Contract]
\qquad \newline
A \emph{committed spend contract} is a
collection of pricing schedules $\{p_\theta\}$ such that each type $\theta$
has a \emph{minimum spend} (sometimes called a \emph{minimum budget})
\begin{equation*}
B(\theta) := \min_{q \in \mathbb{R}_+} \big[p_\theta(q)\big]
\end{equation*}
which is strictly positive for all $\theta > \underline{\theta}$.
\end{definition}

Trivially, any two-part tariff can be converted into a committed spend
contract by combining the upfront payment with the original pricing
schedule. The committed spend contract, because it consists of a single
pricing schedule per type, is always unique (again, up to $p_{\theta }(q)$
which are not chosen in equilibrium). In particular, the committed spend
contract is characterized by 
\begin{equation*}
p_{\theta }(q^{\ast }(\theta ,v))=t(\theta ,v).
\end{equation*}%
The largest possible $B(\theta )$ in an optimal committed spend contract is 
\begin{equation}
B(\theta )=t(\theta ,\underline{v}(\theta )).  \label{eq:b}
\end{equation}%
The ``off-path'' values of $p_{\theta }(q)$
may technically lower the minimum spend, if they are below this value.

This committed spend contract, however, has the feature that the minimum
spend $B(\theta )$ may buy a quantity of 0. That is, the minimum spend is
exactly playing the part of the upfront payment in the two-part tariff, and
does not seem to be a ``true'' minimum
spend contract.

Contracts more similar to what we observe in reality, where the buyer can
always buy some quantity with any positive purchase, are optimal when the
buyer's virtual values are almost always high enough to justify making a
sale.

\begin{proposition}[Guaranteed Positive Quantity]
\label{prop:pos-q} \qquad \newline
There exists a committed spend contract implementing the optimal mechanism
such that, for every $\theta$, 
\begin{equation*}
B(\theta) > 0 \implies \min \Big\{q \mid p_\theta(q) \geq B(\theta)\Big\} > 0
\end{equation*}
if and only if for all $\theta > \underline{\theta}$, 
\begin{equation*}
\varphi(\theta,\underline{v}(\theta)) > 0.
\end{equation*}
\end{proposition}

\begin{proof}
($\Leftarrow$) Since $\varphi(\theta,\underline{v}(\theta)) > 0$, every type 
$\theta > \underline{\theta}$ always receives some allocation, the price of
which forms their largest minimum spend under \eqref{eq:b}. Note that we
exclude $\theta = \underline{\theta}$ from this condition because our
definition of a committed spend contract permits us to set $B(\underline{%
\theta}) = 0$.

($\Rightarrow$) If the condition is not satisfied, then there is no
allocation to $\underline{v}(\theta)$, and then hence by \eqref{eq:b} even
the largest possible minimum spend is not sufficient to purchase any
positive quality.
\end{proof}

This condition is exactly the same as the requirement that the optimal
two-part tariff is non-unique for every $\theta > \underline{\theta}$, as it
exactly involves the situation where the transfer is always strictly
positive (in equilibrium), producing some indeterminacy in how to split the
payments across periods.

Figure \ref{fig:2} plots the (largest) minimum spends $B(\theta)$ for
Example \ref{ex:1}, again with $c = 1$ and $\alpha = \frac{1}{2}$ and the unit price plotted on a parallel axis.

\begin{figure}[ht]
\centering
\includegraphics[height=0.4\textwidth]{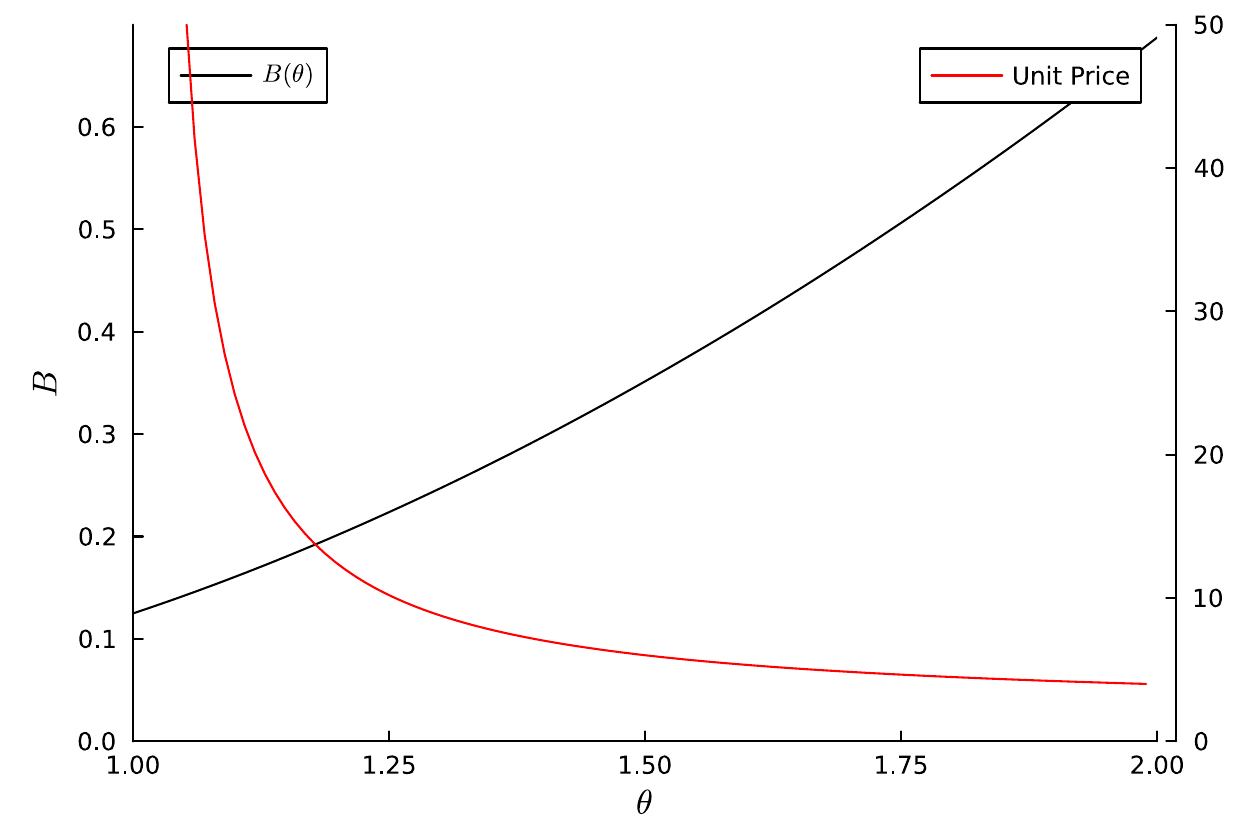}
\caption{$B(\protect\theta)$ and Unit Price for Example \ref{ex:1} when $c = 1$ and $\protect\alpha %
= \frac{1}{2}$.}
\label{fig:2}
\end{figure}

Unlike with the two-part tariff discussed previously, the optimal committed spend
contract is unlikely to exhibit linear pricing, as we would need the average
transfer to remain constant: 
\begin{equation*}
\frac{d}{dq}\left[ \frac{t(\theta ,\theta )}{q(\theta ,v)}\right] =0\iff 
\frac{d}{dq}\Big[t(\theta ,v)\Big]=\frac{t(\theta ,v)}{q(\theta ,v)}.
\end{equation*}%
That is, the marginal transfer must satisfy, for all $v\geq \underline{v}%
(\theta )$, 
\begin{equation*}
\frac{cv}{\varphi (\theta ,v)}=\frac{t(\theta ,\underline{v}(\theta ))}{%
q(\theta ,\underline{v}(\theta ))}=\frac{c}{\alpha }-\left( \frac{c}{\alpha }%
\right) ^{\frac{1}{1-\alpha }}\left[ \frac{u(\theta ,\underline{v}(\theta ))%
}{\underline{v}(\theta )^{\frac{1}{1-\alpha }}}\right] ,
\end{equation*}%
which is a knife-edge case.

\section{Contract Frictions}

In this section, we demonstrate how the solution of the dynamic mechanism is
affected by contractual frictions that arise in cloud computing and software
as a service. First, we introduce a cost of capital for providing payments
in the early period of the relationship. This is particularly relevant in
the cloud computing environment where many early stage companies with large
computing needs are credit constrained. Second, we consider the impact that
a functioning spot-market for computing services has for the designing and
pricing of the dynamic contract.

\subsection{Cost of Committed Capital}

A common reason for the use of a committed spend contract over a two-part
tariff is that early-stage firms, such as startups requiring the use of
computing power, are often very capital-constrained in the early period.
Then, the committed spend contract allows them to effectively make payments
to the seller, but drawn against their future (anticipated) revenues rather
than immediately at the time of contracting.

We can model this as making the timing of transfers matters to the buyer
(but not to the seller). In particular, the buyer incurs a linear penalty $%
\gamma$ for payments made in the early period: 
\begin{equation*}
u(v,q,t) = v q^\alpha - (1+\gamma\mathbb{I}_{[t_0 \geq 0]})t_0 - t_1, \quad
\gamma \geq 0.
\end{equation*}
This is also equivalent to a model with time discounting where the seller is
more patient than the buyer.

The optimal mechanism is still implementable in this framework by
backloading all the transfers to period 1. This exactly selects the
committed spend contract, with minimum spend given by \eqref{eq:b}, as the
optimal contract. In fact, any other mechanism which implements the optimal
contract must involve negative transfers on-path, which is effectively the
seller providing lending to the buyer.

\begin{proposition}[Optimal Contract with Commitment Cost]
\qquad \newline
For any $\gamma > 0$, the committed contract with $B(\theta)$ given by %
\eqref{eq:b} is optimal. Furthermore, it is the only optimal contract where
(on-path) transfers are always positive.
\end{proposition}

Here, there is a social benefit to allowing the buyer to backload payments,
but of course the seller is only willing to allow this if the buyer can
commit to eventually making those payments.

\subsection{Sequential Screening in the Presence of Spot Markets\label%
{sec:spot-market}}

We now consider the interaction between dynamic contract and spot markets.
Suppose that in period 1, buyers have access to a liquid \emph{spot market},
which provides any quantity of the good at constant price $p^{S}>c$. Now,
buyers who sign a dynamic contract in period 0 are bound to it. However,
they anticipate the existence of a future spot market, and do not sign
contracts which deliver expected payoffs worse than participating in the
future spot market. Mathematically, this is equivalent to enhancing the IR
constraint of the buyers. The current setting thus introduces the
possibility that the computing services may be procured through a variety of
sources and market. \cite{stsh21} offer a more comprehensive view of how the
interaction between ``cloud'' and ``sky'' computing may enhance the efficiency
and competitiveness of computing services.

In this extension, we work in the multiplicative setting: that is, for the
committed buyers, gross utility is $\theta z$, where $z\sim H$ is
independent of $\theta \sim F$. The seller's problem is to design a
mechanism $(q,t)$ which maximizes profit subject to \eqref{eq:IC0}-%
\eqref{eq:IC1} and the new IR constraint: 
\begin{equation*}
\E\Big[u(v,q(\theta ,v),t(\theta ,v))\mid \theta \Big]\geq u^{S}(\theta ):=\E%
_{v}\left[ \max_{q}\Big\{vq^{\alpha }-p^{S}q\Big\}\mid \theta \right] ,\
\forall \theta .
\end{equation*}%
We provide a partial characterization of the constrained optimal mechanism, $%
(q^{D},t^{D})$. There exists a cutoff $\theta ^{\ast }$ such that above $%
\theta ^{\ast }$, buyers receive the same allocations as they do in the
original optimal mechanism, $q^{\ast }(\theta ,v)$, but at a discounted
price. For $\theta <\theta ^{\ast }$, the allocation (and prices) are
distorted relative to $(q^{\ast },t^{\ast })$.

\begin{proposition}[Optimal Contract with Spot Market]
\qquad \newline
With multiplicative values, $q^{D}(\theta ,\cdot )=q^{\ast }(\theta ,\cdot )$
if and only if $\theta \geq \theta ^{\ast }$, where $\theta ^{\ast }$ is
determined by 
\begin{equation*}
\frac{\varphi _{F}(\theta ^{\ast })}{\theta ^{\ast }}=\frac{c}{p^{S}}.
\end{equation*}%
Furthermore, when $\theta \geq \theta ^{\ast }$, 
\begin{equation*}
t^{D}(\theta ,v)=t^{\ast }(\theta ,v)-t_{c},
\end{equation*}%
where 
\begin{equation*}
t_{c}=u^{S}(\theta ^{\ast })-u^{\ast }(\theta ^{\ast })\geq 0
\end{equation*}%
is a constant that depends only on $\theta ^{\ast }$.
\end{proposition}

Recall that under the monotone hazard rate assumption (which is Assumption %
\ref{asu:reg} in the multiplicative setting), the ratio of $\varphi
_{F}(\theta )$ to $\theta $ is monotonically increasing, and reaches 1 at $%
\theta =\overline{\theta }$. Thus, $\theta ^{\ast }$ is larger the closer
that $p^{S}$ is to $c$ (i.e.\ the more competitive the spot market is).

So, sufficiently high $\theta$ benefit directly from the spot market, in
terms of a flat discount $t_c$, while low $\theta$ have their allocation
distorted away from the original optimal mechanism. Note that this
distortion may ultimately increase or decrease social efficiency, depending
on the parameters of the problem. This formalizes the notion that having
liquid, functional spot markets is beneficial not only in terms of putting
price pressure on the commitment contract, but also potentially improving
the allocative efficiency.

Formally, the optimal mechanism features no exclusion, meaning that every
buyer participates. However, in practice, there may be (low type) buyers for
whom the optimal mechanism simply replicates the spot market. These buyers
are \emph{de facto} excluded from the mechanism, in that they receive no
per-unit discounts for reporting their type, and are not asked to make any
up-front commitments; in contrast, in the absence of the spot market, all
buyer types sign a non-trivial dynamic contract.

\begin{proof}
First suppose the seller wants to design the optimal contract which captures
the entire market. We modify the original optimal mechanism $(q^{\ast
},t^{\ast })$ in two ways. First, we shift $u(\underline{\theta })$ to be
equal to the outside option, 
\begin{equation*}
t(\theta ,v)=t^{\ast }(\theta ,v)+\int_{\underline{v}}^{\overline{v}%
}u^{S}(v)g(v\mid \underline{\theta })\ dv,
\end{equation*}%
where $u^{S}(v)$ is the net payoff of a buyer of type $v$ participating in
the spot market.

Next, for any values of $\theta$ such that 
\begin{equation*}
\frac{d}{d\theta} \E\big[u^\ast(\theta)\big] < \frac{d}{d\theta} \E\big[%
u^S(\theta)\big],
\end{equation*}
we distort the allocations until this inequality is an equality, subject to
the global \eqref{eq:IC0} constraints. The exact nature of this distortion
is a complicated multi-dimensional mechanism design problem, which is why we
cannot provide a description in general of what the mechanism looks like
when $\theta < \theta^\ast$.

However, we can compute the $\theta ^{\ast }$ above which no distortion is
necessary. Observe that, by the multiplicative structure, 
\begin{equation*}
\frac{\partial G(v\mid \theta )}{\partial \theta }=-\frac{z}{\theta }%
h(z),\quad \varphi (\theta ,v)=\varphi _{F}(\theta )z.
\end{equation*}%
Hence, 
\begin{equation}
\frac{d}{d\theta }\E\big[u^{\ast }(\theta )\big]=-\frac{1}{\theta ^{2}}\int_{%
\underline{z}}^{\overline{z}}\left( \frac{\alpha }{c}\varphi _{F}(\theta
)z\right) ^{\frac{\alpha }{1-\alpha }}zh(z)\ dz.  \label{eq:dt-1}
\end{equation}%
The spot market clearly defines a dynamic mechanism over $(\theta ,v)$ which
satisfies all IC constraints (including the local ones), so we can apply %
\eqref{eq:u1} to it as well. In particular, a buyer with realized type $%
v=\theta z$ participating in the spot market purchases 
\begin{equation*}
q(\theta ,z)=\left( \frac{\alpha z\theta }{p^{S}}\right) ^{\frac{1}{1-\alpha 
}}
\end{equation*}%
and thus
\begin{equation}\label{eq:dt-2}
    \frac{d}{d\theta }\E\big[u^{S}(\theta )\big]=-\frac{1}{%
\theta ^{2}}\int_{\underline{z}}^{\overline{z}}\left( \frac{\alpha z\theta }{%
p^{S}}\right) ^{\frac{\alpha }{1-\alpha }}zh(z)\ dz.
\end{equation}%
The ratio of \eqref{eq:dt-1} to \eqref{eq:dt-2} is 
\begin{equation*}
\left( \frac{p^{S}}{c}\cdot \frac{\varphi _{F}(\theta )}{\theta }\right) ^{%
\frac{\alpha }{1-\alpha }}
\end{equation*}%
By assumption, this is increasing in $\theta $, so the necessary distortion
is declining with $\theta $. Hence, once 
\begin{equation*}
\frac{\varphi _{F}(\theta )}{\theta }=\frac{c}{p^{S}},
\end{equation*}%
distortion is no longer necessary. Since above $\theta ^{\ast }$, the
allocations are equal to $q^{\ast }$, the transfers must also be equal to $%
t^{\ast }$ plus a constant, which is the constant necessary to make the IR
constraint binding at $\theta ^{\ast }$. Below $\theta ^{\ast }$, at least
some distortion is necessary, or the IR constraint would not be met.

Finally, no exclusion is ever optimal since the seller can at worst
replicate the spot market for low types, and with $p^S > c$ this is always
profitable.
\end{proof}

It is apparent from the upper bound on $\theta^\ast$ that a more competitive
spot market makes the distortion relative to $q^\ast$ more severe, and in
the limit where $p^S = c$ the seller of course cannot do anything except
replicate the spot market.

\section{Conclusion}

There is a growing regulatory interest in the potentially anti-competitive
effects of committed spend contracts. Our model makes clear that the ability
to force buyers to commit to future payment schemes obviously benefits the
seller, but also benefits high type buyers, at the expense of low-type ones.
We could analyze the overall benefit of allowing these contracts by
comparing with the optimal mechanism under ex-post implementability
constraints (which captures the lack of commitment from the buyer),
extending work such as \cite{chel02} and \cite{becw20}.

The setup of Section \ref{sec:spot-market} also leaves open the possibility of
modeling more explicitly competition in the spot market, as well what
happens when the seller is also the designer of the spot market, which is
common in practice. There, we expect that the seller would often find it
optimal to degrade the spot market in order to extract additional rents from
the commitment contract.

Our analysis has several implications for ongoing policy debates about
committed spend agreements. First, it suggests that while committed spend
contracts can improve efficiency by enabling better capacity planning and
reducing uncertainty, they may indeed create switching costs that reduce
ex-post competition. However, the welfare impacts depend crucially on market
structure---in monopolistic markets, commitments mainly affect rent
extraction rather than efficiency. Second, committed spend contracts can be
particularly valuable when capital markets are imperfect, as they allow
resource-constrained firms to effectively borrow against future service
usage. Policy interventions should consider these financial market
interactions. Third, the presence of liquid spot markets can discipline
long-term contract prices and improve overall market efficiency. Regulators
might focus on ensuring robust spot market development rather than directly
restricting committed spend agreements. Finally, our model suggests that
optimal commitment periods should be linked to the precision of demand
signals. This provides a framework for evaluating whether particular
contract durations are anticompetitive.

\newpage

\bibliographystyle{econometrica}
\bibliography{general}

\end{document}